\documentclass[11pt]{article}
\usepackage{fullpage}
\usepackage{algorithmic}
\usepackage{epsfig}
\usepackage{epstopdf}
\usepackage{graphicx}
\usepackage{latexsym}
\usepackage{amsmath}
\usepackage{amsfonts}
\usepackage{multirow}
\usepackage{bigstrut,multirow,rotating}
\usepackage{amssymb}

\usepackage{mathrsfs}
\usepackage{pifont}
\usepackage{yhmath}
\usepackage{undertilde}
\usepackage{bbm}
\usepackage{eufrak}
\usepackage{amsthm}
\usepackage{booktabs}
\usepackage[usenames,dvipsnames]{color}
\usepackage{stmaryrd}
\usepackage{wasysym}
\usepackage{bbm}
\usepackage{array}
\usepackage{enumerate}

\newtheorem{theorem}{Theorem}[section]
\newtheorem{lemma}[theorem]{Lemma}

\newtheorem{proposition}[theorem]{Proposition}

\theoremstyle{remark}

\makeatletter
\newcommand\figcaption{\def\@captype{figure}\caption}
\newcommand\tabcaption{\def\@captype{table}\caption}
\makeatother

\newcommand{\comment}[1]{\textcolor{black}{\textrm{#1}}}

\DeclareMathAlphabet{\mathpzc}{OT1}{pzc}{m}{it}

\begin{document}
\newcounter{my}
\newenvironment{mylabel}
{
\begin{list}{(\roman{my})}{
\setlength{\parsep}{-1mm}
\setlength{\labelwidth}{8mm}
\usecounter{my}}
}{\end{list}}

\newcounter{my2}
\newenvironment{mylabel2}
{
\begin{list}{(\alph{my2})}{
\setlength{\parsep}{-0mm} \setlength{\labelwidth}{8mm}
\setlength{\leftmargin}{3mm}
\usecounter{my2}}
}{\end{list}}

\newcounter{my3}
\newenvironment{mylabel3}
{
\begin{list}{(\alph{my3})}{
\setlength{\parsep}{-1mm}
\setlength{\labelwidth}{8mm}
\setlength{\leftmargin}{10mm}
\usecounter{my3}}
}{\end{list}}

\title{\bf Favorite-Candidate Voting for Eliminating the Least Popular Candidate in Metric Spaces}
%\thanks{Research supported in part by NNSF of China under Grant No. 11531014.}

\maketitle

\vspace{-2em}

\begin{center}

\author{Xujin Chen,\textsuperscript{\rm 1,2}
Minming Li,\textsuperscript{\rm 3}
Chenhao Wang\textsuperscript{\rm 1,2,3}\\
\medskip\textsuperscript{\rm 1}Academy of Mathematics and Systems Science, Chinese Academy of Sciences\\
\textsuperscript{\rm 2}School of Mathematical Sciences, University of Chinese Academy of Sciences\\
\textsuperscript{\rm 3}Department of Computer Science, City University of Hong Kong\\
xchen@amss.ac.cn,\;\; minming.li@cityu.edu.hk,\;\; chenhwang4-c@my.cityu.edu.hk
}

\end{center}
%\vspace{2em}

%\begin{abstract}
%
%
% \noindent{\bf Keywords:}
%\end{abstract}
%
\begin{abstract}
We study single-candidate voting embedded in a metric space, where both voters and candidates are points in the space, and the distances between voters and candidates specify the voters' preferences over candidates. In the voting, each voter is asked to submit her {\em favorite} candidate. Given the collection of favorite candidates, a mechanism for eliminating the least popular candidate finds a  {\em committee} containing all candidates but the one to be eliminated.

  Each committee is associated with a social value that is the sum of the costs (utilities) it imposes (provides) to the voters. We design mechanisms for finding a committee to optimize the social value. We measure the quality of a mechanism by its \emph{distortion}, defined as the worst-case ratio between the social value of the committee found by the mechanism and the optimal one. We establish new upper and lower bounds on the distortion of  mechanisms in this single-candidate voting, for both general metrics and well-motivated special cases.
\end{abstract}

\section{Introduction}
\hspace{3mm}
In social choice theory, a \emph{mechanism} (also referred to as a \emph{voting rule}) aggregates the preferences of multiple voters over a set of candidates, and returns  a $k$-element subset of candidates as a winning committee. %The preference of a voter is often represented as a linear order over the candidates or the top-ranked candidate. There is no consensus on what are the best voting rules, and most of the social choice literature compares different voting rules by defining normative or axiomatic criteria, such as anonymity, neutrality and Pareto efficiency.
An appealing approach to dealing with social choice problems is embedding the input ``election'' into a \emph{metric space}, \emph{i.e.}, each participant is represented by a point in a metric space, and voters %tend to
prefer candidates that are closer to them to the ones that are further away. This spatial model has very natural interpretations. For example, in a 2-dimensional Euclidean space, each dimension specifies a political issue (such as military or education), and the position of a voter or candidate identifies the extent to which the individual supports the issues. Recently, this model has attracted attentions from AI researchers, see, {\em e.g.,} \cite{anshelevich2015approximating,elkind2017multiwinner,abramowitz2019awareness}.

The mechanisms in many of the works aforementioned ask each voter for a linear order over candidates. On the other hand, one may note that eliciting so much information on the preferences casts a high burden on the selection rules, and often impairs the privacy of voters. %Therefore, the
The \emph{simplicity}, which means that each voter is only required to provide a small amount of information, is often %becomes
a desideratum for good mechanisms. In this paper, we study the \emph{single-candidate vote  mechanisms} (named by \cite{feldman2016voting}), {\em scv mechanisms} for short, that ask each voter to cast a vote of a single candidate.

In addition to the top choices of voters, we further assume that the locations of candidates in the metric space are known to the mechanism, while the voters' \comment{private} locations and numerical preferences are inaccessible, since %For the reasons,
 every political candidate in a typical election should fully announce her opinions on all issues, and thus her location in the space is public information. \comment{For example, in} the facility location scenario, the city authority, who plans to locate some facilities on a street or a plane, predetermines the potential locations of facilities, based on the landscape, resources and distributions of social communities.

 % Foremost, the underlying metric space represents the \comment{external} environment of the election, that is often the public information to all participates.

 As voters' preferences are specified by their distances to candidates, % When the preference of a voter is specified by her distances to each candidate,
 it is natural to quantify the {\em quality} of a committee by the associated distances. %the sum of its distances to all voters. We aim at selecting a committee that optimizes (minimizes or maximizes) this objective function.
  We evaluate the  performance of a mechanism in the standard worst-case analysis benchmark (introduced by Procaccia and Rosenschein \cite{procaccia2006distortion}), which defines the \emph{distortion} of a mechanism to be the worst-case ratio between the quality of a committee selected by this mechanism and that of the optimal committee selected by an omniscient mechanism.
 % More generally, the distortion, which stems from an information deficiency, is a measure of how close it always gets to the optimal alternative.

   Previous work was mainly concerned about the single-winner elections. In this paper, we focus on the antithesis, the multi-winner elections that eliminate the least popular candidate, that is, select a committee containing all candidates but one. These can be regarded as \comment{\emph{single-loser} elections, which are} well motivated. For example, some enterprises adopt a \emph{last-out} mechanism in the personnel performance appraisal system, which dismisses the employee with the lowest performance in a department each year. Some voting rules in TV talent shows iteratively eliminate one candidate at each time to obtain the final winners. %a voting rule for the single-loser election may be iteratively used as a subroutine for multi-winner elections, eliminating one candidate at each iteration.

 %Previous work mainly concerned about the social cost objective, providing distortion bounds for scv mechanisms in single-winner elections and \emph{ranking} mechanisms that ask for linear orders as input. In this paper, we bound the distortion for scv mechanisms in multi-winner elections w.r.t. social cost, and scv mechanisms in single-winner elections w.r.t. social utility.

\vspace{2mm}\noindent\textbf{Our Contributions.}

Let $m$ be the number of candidates in the election, and $W$ be the winning committee of size $m-1$ selected by a mechanism. We discuss the distortion of mechanisms under two objectives: minimizing the \emph{social cost} and maximizing the \emph{social utility}. In the former case, each voter takes the distance to $W$ (\emph{i.e.}, the smallest distance between her and a candidate in $W$) as her cost, and the social cost of $W$ is the sum of its distances to all voters. In the latter case, each voter takes her distance to the eliminated candidate (\emph{i.e.}, the one not in $W$) as her utility, and the social utility of $W$ is the total \comment{utility} of voters.

In Section \ref{sec:sc}, we study the distortion of scv mechanisms under the social cost objective.
We prove that if \comment{the exact locations of the candidates are known,} % we know the exact locations of the candidates in addition to the favorite candidate of the voters,
 then a simple deterministic mechanism which minimizes the so-called projection distance achieves a distortion of 3, and no deterministic one can do better.  In
other words, we can compute a 3-approximate solution \comment{as long as the input votes are consistent with the true distances, \emph{i.e.,} each vote is indeed a candidate closest to the voter.} %no matter what the true distances from voters to candidates are %, as long as they are consistent with the  top-choice preferences.
The most interesting contribution is a randomized scv mechanism with distortion $3-\frac2 m$,  which selects each eligible committee with a carefully designed probability. We prove that no randomized mechanism has a distortion better than $3-\frac 2m$, matching the upper bound. The deterministic and randomized mechanisms also satisfy \emph{strategy-proofness}, guaranteeing %which requires mechanisms to guarantee
that each selfish voter always acts truthfully.  Moreover, the  lower bounds $3$ and $3-\frac2m$ hold even if the voters submit a full preference ranking over all the candidates.

 Section \ref{sec:su} focuses on the social utility objective.
We show the \comment{lower bounds 3 and 1.5 for deterministic and randomized mechanisms}, respectively.
While the deterministic mechanism that maximizes the projection distance gives a distortion 3 for general metrics, we investigate randomized scv mechanisms for elections in several widely-studied special spaces, \emph{e.g.}, the simplex (where the distance between any two candidates is the same) and the real line (1-Euclidean space). The simplex setting corresponds to the case when
candidates share no similarities, i.e., when all candidates are equally different from each other, and the real line is also a well-studied and well-motivated special case.

These results are summarized in Table \ref{table}, where LB and UB are shorthands for lower bound and upper bound on the distortion of scv mechanisms.
%\vspace{-2mm}
\begin{table}[htbp]
  \centering
  \caption{A summary of our results}\label{table}

    \begin{tabular}{|p{2cm}<{\centering}|p{3cm}<{\centering}|p{5cm}<{\centering}|}
    \hline
    \small Objective &  \small Deterministic &\small Randomized \bigstrut\\
    \hline
    \multirow{2}[2]{*}{\small Social cost} & \multirow{1}[1]{*}{\small LB: ~$3$ (Prop.\ref{prop:poplower1}) } & {\small LB: $3-\frac2 m$ (Prop.\ref{prop:poplower})} \bigstrut[t]\\

                                    & {\small UB: $3$ (Thm.\ref{th:min})} & {\small UB: $3-\frac2 m$ (Thm.\ref{thm:pow})} \bigstrut[b]\\
    \hline
    \multirow{3}[2]{*}{\small{Social utility}} & \multirow{2}[1]{*}{\small LB: $3$ (Prop.\ref{prop:lower})} & {\small LB: 1.5 (Prop.\ref{prop:lower})} \bigstrut[t]\\

                                    &  {\small UB: $3$ (Thm.\ref{th:max})} & {\shortstack{\small UB: ~$3$-$\frac{4}{m+2}$ (Simplex, Thm.\ref{thm:pmech})\\\small $13/7$ (Line, Thm.\ref{thm:line})}} \bigstrut[b]\\
    \hline
    \end{tabular}
  \label{tab:1}
\end{table}

In Section \ref{sec:add}, we extend our results to a more general setting, where the scv mechanism is required to select a committee of size $k$, for a predetermined integer $k\le m-1$. We prove that the simple idea that optimizes the projection distance can achieve a distortion 3 for both the social cost objective and the social utility objective, and no deterministic mechanism can do better.  Then we conclude this paper with future research directions.

\vspace{2mm}\noindent\textbf{Related Work.}

In social choice theory, Procaccia and Rosenschein \cite{procaccia2006distortion} propose a utilitarian approach \comment{-- the implicit utilitarian voting--} by assuming that voters have latent cardinal utilities and report ordinal preferences induced by them. They measure the performance of popular voting rules by the notion of \emph{distortion}. Subsequently, Caragiannis and Procaccia \cite{caragiannis2011voting}, Oren and Lucier \cite{oren2014online}, Boutilier \emph{et al.} \cite{boutilier2015optimal}, Bhaskar and
Ghosh \cite{bhaskar_et_al:LIPIcs:2018:9926} employ this notion and design selection rules with low distortions.

 Anshelevich {\em et al.} \cite{anshelevich2015approximating} first embed the election into a metric space, in which the participants are points, and the costs are driven by the distances. \comment{They study mechanisms  that know only the voters' preference rankings} %has knowledge of only the voters preference ranking
 over candidates, but not the underlying metric, and output a single winner. \comment{Regarding the objective of minimizing the social cost of the winner, they}
 show the Copeland rule has distortion 5, and \comment{prove} a lower bound 3 for the distortion of deterministic \comment{mechanisms}. Later, Skowron and Elkind \cite{skowron2017social} show that the class of scoring rules and STV have super-constant distortion. The work of \cite{goel2017metric} proves that the ranking pairs rule has distortion at least 5. Recently,  Munagala and Wang \cite{munagala2019improved} improve the distortion to 4.236, using a weighted tournament rule.

 \comment{In addition to} %While the above works only focus on
 deterministic rules, randomized rules have also been considered. Random dictatorship that randomly selects the top choice of one of the $n$ voters gets distortion $3-2/n$ \cite{feldman2016voting,anshelevich2017randomized}. Feldman {\em et al.} \cite{feldman2016voting} consider scv mechanisms and strategy-proofness in the metric setting, and propose a 2-distortion mechanism on the real line.  The work of \cite{gross2017vote} proposes \comment{a very simple} mechanism %of simplicity
 that randomly asks voters for their favorite candidates until two voters agree, achieving low distortion and satisfying some normative properties.

 The most related \comment{setting to ours appears in} %the work of
 \cite{anshelevich2018ordinal}, where the candidates' locations are additionally assumed to be known. %that additionally assumes the locations of candidates are known.
 With this extra location information,  they break the best-known upper bound 4.236 mentioned above and present a deterministic 3-distortion scv mechanism for single-winner election.

\section{Model}\label{model1}

Let $\Omega=(S,d)$ be a metric space, where $S$ is the space and $d:S\times S\rightarrow \mathbb R_+$ is the metric. The {\em distance} between $w\in S$ and $V\subseteq S$ is defined as $d(w,V):=\min_{v\in V}d(w,v)$. Let $N=\{1,\ldots,n\}$ be the set of \emph{voters} (agents), each of whom is located at a private point in $S$. The location $x_i$ of voter $i\in N$ is her \emph{type}, and the \emph{location profile} of all voters is $\mathbf x=(x_1,\ldots,x_n)$.  Let $M=\{y_1,\ldots,y_m\}$ be the set of \emph{candidates} (alternatives), each of whom is located at a public point in $S$. We refer to $y_j$ as the $j$-th candidate and as her location interchangeably.

The voter prefers the closer candidate, %than the one farther away,
 and the nearest candidate is the favorite. Each voter $i\in N$ is asked to submit a single nearest candidate, called her \emph{action} and denoted by $a_i\in M$. The collection of voters' actions is the \emph{action profile} $\mathbf a=(a_1,\ldots,a_n)$. An \emph{election} in the social choice problem under consideration is a triple $\Gamma=(\Omega,M,\mathbf a)$. We call a location profile $\mathbf x$ \emph{consistent} with election $\Gamma$, if each voter's action reveals her real preference, that is, {$a_i\in\arg\min_{y\in M}d(x_i,y)$,} for every $i\in N$. Denote by $\chi(\Gamma)$ the set of location profiles consistent with $\Gamma$.

We are concerned with \emph{mechanisms} that, given an election $\Gamma=(\Omega,M,\mathbf a)$, select a \emph{committee} (subset of $M$) of cardinality $m-1$ as winners. It is assumed that the mechanisms have full information on the metric space $\Omega$ and  candidate locations $M$, but they do not know the location profile of voters. \comment{Associate each $y\in M$ with the potential committee $M_y:=M\setminus\{y\}$. Let  $K=\{M_y:y\in M\}$ denote the set of potential committees.}  A \emph{randomized} mechanism is a function $f$ that maps every action profile $\mathbf a\in M^n$ to a random committee $f(\mathbf a)$ that follows some probability distribution over the potential committees in $K$. A \emph{deterministic} mechanism $f$ simply selects  a specific committee $f(\mathbf a)\in K$ with probability 1.

We investigate the performance of mechanisms from the utilitarian perspective, which involves the objectives of minimizing the social cost and maximizing the social utility, respectively.

\vspace{2mm}\textbf{The social cost objective.}  %Each voter takes the distance to the nearest winner as her \emph{cost}, that is,
Given location profile $\mathbf x=(x_i)_{i\in N}$ and committee $Y\in K$, the {\em cost} of voter $i\in N$ is the distance to the nearest winner, \emph{i.e.,} $d(x_i,Y)$. The {\em social cost} of $Y$, denoted as $SC(Y,\mathbf x)$ or $SC(Y)$ for short, equals $\sum_{i\in N}d(x_i,Y)$. We use $OPT_c(\mathbf x)$ to denote the social cost of an optimal committee selected by an omniscient mechanism, \emph{i.e.}, $OPT_c(\mathbf x)=\min_{Y\in K}SC(Y,\mathbf x)$. The \emph{distortion} of a (randomized) mechanism $f$ on an election $\Gamma=(\Omega,M,\mathbf a)$ is
 $$dist(f,\Gamma)=\sup_{\mathbf x\in \chi(\Gamma)}\frac{\mathbf E[SC(f(\mathbf a),\mathbf x)]}{OPT_c(\mathbf x)}.$$
 In other words, it is the worst-case --- over the location profiles consistent with $\Gamma$ %the given action profile
 --- ratio between the expected social cost of the committee selected by the mechanism and the optimal social cost.

\textbf{The social utility objective.} %Each voter takes the distance to the loser as her \emph{utility}, that is,
Given location profile $\mathbf x$ and committee $M_y$, the utility of voter $i$ equals \comment{the $d(x_i,y)$ to the loser $y$}. The {\em social utility} of $M_y$, denoted as $SU(M_y,\mathbf x)$ or $SU(M_y)$ for short, equals $\sum_{i\in N}d(x_i,y)$. The optimal social utility is % defined as
$OPT_u(\mathbf x)=\max_{Y\in K}SU(Y,\mathbf x)$, and the \emph{distortion} of a (randomized) mechanism $f$ on election $\Gamma$ is
 $$dist(f,\Gamma)=\sup_{\mathbf x\in \chi(\Gamma)}\frac{OPT_u(\mathbf x)}{\mathbf E[SU(f(\mathbf a),\mathbf x)]}.$$

For either of the objectives, we define the {\em distortion} of a mechanism $f$ as $Dist(f)=\sup_{\Gamma}dist(f,\Gamma)$ by taking the worst case over elections. We call $f$ an {\em $r$-distortion mechanism} if $Dist(f)\le r$.

 \vspace{2mm}\textbf{Strategy-proofness.} As in many previous works on social choice, we evaluate the quality of a mechanism under the assumption that the {underlying} location profile is always consistent with the elections, \emph{i.e.}, the voters act truthfully and submit their nearest candidates. Nevertheless, possibly some voter may use a strategy (that leads to an action and consequently an election with which the location profile may not be consistent) to be better off. A mechanism is \emph{strategy-proof}, if the truth-telling strategy is always optimal for each voter, that is, voting for {\em any} one of the nearest candidates can {\em always} optimize her (expected) cost or utility, regardless of the actions of others.%\footnote{\comment{In this paper, we only investigate the strategy-proofness of mechanisms rather than the stronger notion of  group strategy-proofness.}}

\section{Mechanisms for Minimum Social Cost}\label{sec:sc}

This section focuses on the objective of minimizing the social cost. We first show the lower bounds on distortion, and propose both deterministic and randomized mechanisms that match the lower bounds.

\subsection{Lower Bounds}\label{sec:sizek}

We prove lower bounds on the distortion of both deterministic and randomized mechanisms by constructing election instances. Our construction is based on the well-known worst case \comment{of} single-winner election \cite{anshelevich2015approximating,feldman2016voting}, in which two candidates locate at 0 and 2 on the real line respectively, and \comment{each receive a vote.} %two candidates each receiving a vote locate at 0 and 2 on the real line, respectively.
We extend it to our setting by adding $m-2$ very far candidates,  each of whom also receives a vote. Then any mechanism with guaranteed performance must weed out \comment{either the candidate  locating at 0 or the one at 2; while} %one candidate between 0 and 2, but
either option results in a distortion 3.

\begin{proposition}\label{prop:poplower1}
For any $m\ge2$ and the social cost objective, the distortion of any deterministic scv mechanism cannot be smaller than~$3$. %\comment{In addition, this lower bound holds for any $m\ge 2$.}
\end{proposition}
\begin{proof}
Consider an election $\Gamma$ in $\mathbb R$, where $m$ candidates are located at $y_1=0,y_2=2,y_3=L,y_4=2L,\ldots,y_m=(m-2)L$ for a large number $L$, and the action profile of $n=m$ voters is $\mathbf a=(0,2,L,2L,\ldots,(m-2)L)$.

It is easy to see that any mechanism $f$ with bounded distortion must eliminate either $y_1$ or $y_2$. If $y_1\in f(\mathbf a)$, then for the location profile $\mathbf x=(1,2,L,2L,\ldots,(m-2)L)\in\chi(\Gamma)$, we have $SC(f(\mathbf a),\mathbf x)=3$, and  $OPT_c(\mathbf x)=1$ (realized by the optimal committee $M_{y_1}$), indicating the distortion at least 3. If $y_2\in f(\mathbf a)$, the same bound holds for location profile $(0,1,L,2L,\ldots,(m-2)L)$.
\end{proof}

Although the  example constructed above can provide a lower bound 2 for the distortion of randomized scv mechanism, we prove a better lower bound in the following.

\begin{proposition}\label{prop:poplower}
For the social cost objective, the distortion of any randomized scv mechanism cannot be smaller than $3-\frac 2 m$.
\end{proposition}
\begin{proof}
Consider an election \comment{$\Gamma=(\Omega,M,\mathbf a)$ with} $d(y_i,y_j)=2$ for any \comment{pair of distinct} candidates $y_i,y_j\in M$. There are $n=m$ \comment{voters}, and the action profile is $\mathbf a=(y_1,y_2,\ldots,y_m)$, that is, each candidate receives a vote from one \comment{voter. Since there are in total $m$ potential committees, any} randomized mechanism $f$ must \comment{select} some committee $M_y$ with a probability no more than $\frac 1 m$. \comment{By symmetry,} % as there are totally . Since the actions received and candidates' locations are \comment{indistinctive},
we can assume w.l.o.g. that \comment{$\mathbf{Pr}[f(\mathbf a)=M_{y_m}]\le\frac1m$.} %$Y=\{y_1,y_2,\ldots,y_{m-1}\}$ and $P(Y)\le \frac1 m$.

Now consider the location profile $\mathbf x=(y_1,y_2,\ldots,y_{m-1},x_m)$, where the point $x_m$ is at the same distance $d(x_m,y_i)=1$ from every candidate $y_i\in M$.   Obviously, suitable choice of $\Omega,M $ and $x_m$ can fulfill all the conditions ({\em i.e.}, the distances specified satisfy the metric \comment{condition}), and guarantees that $\mathbf x$ is consistent with $\Gamma$. (Figure \ref{fig:simplex} depicts an example \comment{for} $m=3$.)
The optimal committee is $M_{y_m}$ \comment{with optimal social cost} $OPT_c(\mathbf x)=d(x_m,M_{y_m})=1$\comment{, while any other committee $M_{y_i}$ with $i\le m-1$ has a social cost at least $d(x_m,M_{y_i})+d(y_i,M_{y_i})=1+2=3$. Thus, the expected social cost of the random committee $f(\mathbf a)$   is  $\mathbf{E}[SC(f(\mathbf a),\mathbf x)]\ge \frac1 mOPT_c(\mathbf x)+(1-\frac1 m)3=3-\frac2 m$, showing that  the distortion of $f$} is at least $3-\frac 2 m$.
\end{proof}

\vspace{-5mm}
\begin{figure}[htpb]
\begin{center}
\includegraphics[scale=0.5]{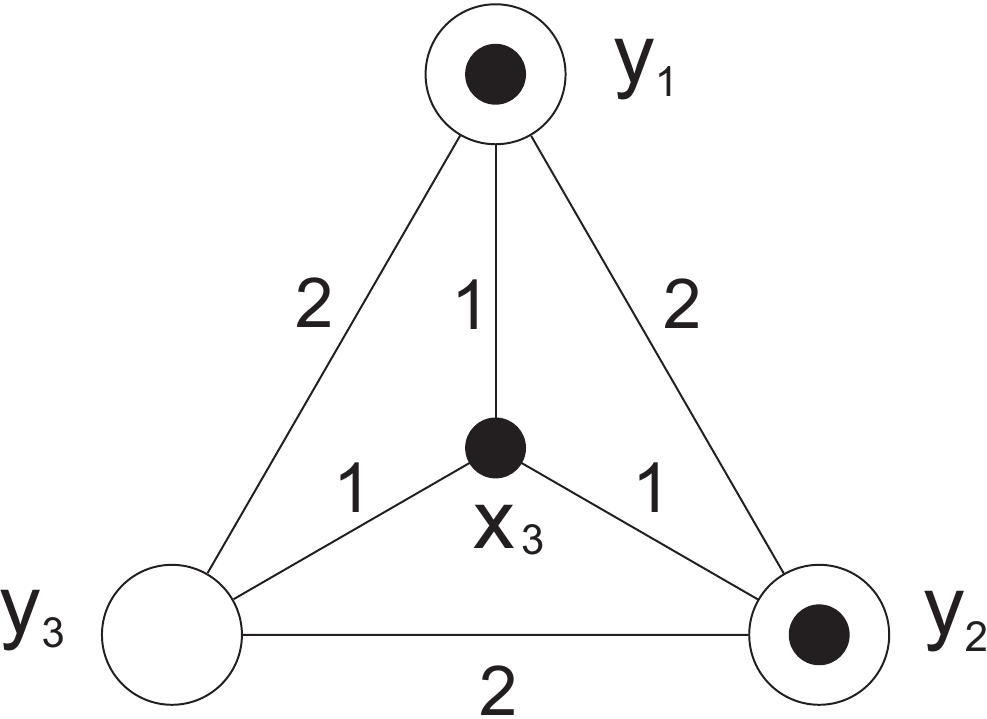}%\includegraphics[scale=0.4]{pic.eps}
\caption{\label{fig:simplex}Three candidates are indicated by hollow circles,  three voters are indicated by solid disks. Each candidate receives a vote. The numbers near edges indicate the distances, which are not Euclidean. The point $x_3$ is at distance 1 from every candidate, and the optimal solution eliminates $y_3$.}
\end{center}
\end{figure}

\comment{It is worth pointing out that the two election examples constructed in the proofs of Propositions \ref{prop:poplower1} and \ref{prop:poplower}} can be applied to the election that asks each voter to submit a preference ranking. Thus the lower bounds  {in these two proportions} also hold for \comment{the mechanisms that aggregate voters' rankings over candidates}.

\subsection{Projection Mechanism}

Given an action profile $\mathbf a=(a_i)_{i\in N}$, \comment{it can be viewed as a projection of the location profile of voters to the location set of candidates. For any subset $W\subseteq M$, we define its \emph{projection distance} w.r.t. $\mathbf a$ as}
$pd_{\mathbf a}(W):=\sum_{i\in N}d(a_i,W)$. In the remainder of this paper, we use $d(i,V)$ instead of $d(x_i,V)$ for $i\in N$ and $V\subseteq S$, when the context is clear. Now we are ready to present a deterministic mechanism which ensures  distortion 3  matching the lower bound in Proposition \ref{prop:poplower1}, by selecting a committee that minimizes the projection distance. \\

\hspace{-3mm}\textbf{Mechanism 1 \comment{\sc(Min-Projection-Distance)}.} Given an election $\Gamma=(\Omega,M,\mathbf a)$,  mechanism $f$ deterministically outputs \comment{a committee} $f(\mathbf a)$ with the smallest projection distance, that is, $f(\mathbf a)\in\mathrm{argmin}_{M_y:y\in M}pd_{\mathbf a}(M_y)$; ties are broken  arbitrarily.\\

The \comment{spirit} of this mechanism is treating the action of each voter as her location.

\begin{theorem}\label{th:min}
\comment{\sc Min-Projection-Distance} is a deterministic, {strategy-proof, polynomial-time and} $3$-distortion scv mechanism for the social cost objective.
\end{theorem}
\begin{proof}
 The polynomial-time computability is straightforward since the number of possible committees is $|K|=m$.

For the strategy-proofness,  we show that the truth-telling strategy always gives each voter a minimum cost. Suppose action $a_i$ is a nearest candidate of voter $i$, and $a_i'\in M\setminus\{a_i\}$ is another arbitrary action. Given the actions $\mathbf a_{-i}$ of other voters, consider the action profiles $\mathbf a=(a_i,\mathbf a_{-i})$ and $\mathbf a'=(a_i',\mathbf a_{-i})$. The output of the mechanism is $f(\mathbf a)=Y\in\mathrm{argmin}_{M_y}pd_{\mathbf a}(M_y)$ and $f(\mathbf a')=Y'\in\mathrm{argmin}_{M_y}pd_{\mathbf a'}(M_y)$. We only need to consider the case where $Y\neq Y'$. If $a_i\in Y$, the cost of voter $i$ is minimized when she tells the truth. So we assume  $a_i\notin Y$, which along with $k=m-1$ implies $a_i\in Y'$. If $a_i'\in Y'$, the projection distance of $Y$ on $\mathbf a'$ is \comment{$pd_{\mathbf a'}(Y)=pd_{\mathbf a}(Y)-d(a_i,Y)<pd_{\mathbf a}(Y)\le pd_{\mathbf a}(Y')$}, and the projection distance of $Y'$ on $\mathbf a'$ is $pd_{\mathbf a'}(Y')=pd_{\mathbf a}(Y')$ since both $a_i$ and $a_i'$ are in $Y'$. So we have $pd_{\mathbf a'}(Y)<pd_{\mathbf a'}(Y')$, which contradicts the \comment{selection} rule of the mechanism, and \comment{reduces to the case of   $a_i\in Y'\setminus Y$ and $a_i'\in Y\setminus Y'$. Now} we have $pd_{\mathbf a'}(Y')>pd_{\mathbf a}(Y')\ge pd_{\mathbf a}(Y)> pd_{\mathbf a'}(Y)$, which also contradicts the \comment{selection} rule. Therefore, voter \comment{$i$'s} cost when reporting $a_i$ is always no more than \comment{her} cost when reporting any $a_i'$, which proves the strategy-proofness.

Given election $\Gamma$ and any consistent location profile $\mathbf x=(x_i)_{i\in N}\in\chi(\Gamma)$, let $Y^*$ be the optimal committee, and $Y$ be the output by the mechanism. Then %the ratio between the social cost of $Y$ and $X$ is
\begin{align*}
\frac{SC(Y)}{SC(Y^*)}&=\frac{\sum_{i\in N}d(x_i,Y)}{\sum_{i\in N}d(x_i,Y^*)}\\
&\le\frac{\sum_{i\in N}d(x_i,a_i)}{\sum_{i\in N}d(x_i,Y^*)}+\frac{\sum_{i\in N}d(a_i,Y)}{\sum_{i\in N}d(x_i,Y^*)}\\
&\le 1+\frac{\sum_{i\in N}d(a_i,Y)}{\sum_{i\in N}d(x_i,Y^*)}.
\end{align*}
For every $i\in N$, recalling from the consistency that $d(x_i,a_i)=\min_{y\in M}d(x_i,y)\le d(x_i,Y^*)$, we have  {$2d(x_i,Y^*)\ge d(x_i,Y^*)+d(x_i,a_i)\ge d(a_i,Y^*)$}. Therefore
$$\frac{SC(Y)}{SC(Y^*)}\le 1+\frac{\sum_{i\in N}d(a_i,Y)}{\frac12\sum_{i\in N}d(a_i,Y^*)}=1+\frac{2pd_{\mathbf a}(Y)}{pd_{\mathbf a}(Y^*)}\le 3,$$
where the last inequality is guaranteed by the selection rule of the mechanism.
\end{proof}

%--------------------------------------------------------------

\subsection{Power-Proportionality Mechanism}\label{sec:sizem-1}

 Inspired by \cite{anshelevich2017randomized}, we establish in the following,   for any given randomized scv mechanism and  location profile, an upper  bound on the ratio between the expected social cost of the committee selected by the mechanism, and the optimal social cost. With the help of this \comment{upper} bound, we design a randomized \comment{scv} mechanism, and prove its strategy-proofness and distortion  \comment{(which matches} the lower bound in Proposition \ref{prop:poplower}).

 Before presenting %proceeding to
  the formal description of the upper bound, we make a partition of the voter set according to voters' actions. Given an action profile \comment{$\mathbf a=(a_i)_{i\in N}$}, for each candidate $y\in M$, \comment{let $N_{\mathbf a,y}=\{i\in N|a_i=y\}$ denote}   the subset of \comment{voters whose actions are  $y$. Then $(N_{\mathbf a,y})_{y\in M}$ forms a partition of $N$.}

 \begin{lemma}\label{lem:simi}
Given a randomized scv mechanism and an election $\Gamma=(\Omega,M,\mathbf a)$, % with $k=m-1$,
suppose the probability that the mechanism selects each $M_{y}\in K$ as winners is $P(M_{y})$. Then, for any %(\comment{not necessarily consistent})
  location profile $\mathbf x\in\chi(\Gamma)$ and  any optimal committee $M_{y^*}$ with $y^*\in M$, the following holds:
   \begin{align}
& \frac{\sum_{y\in M}P(M_{y})SC(M_{y})}{SC(M_{y^*})} \label{eq:disss}\\
\le& 1+\frac{2\sum_{y\neq y^*}P(M_{y})|N_{\mathbf a,y}|d(y,M_{y})}{|N_{\mathbf a,y^*}|d(y^*,M_{y^*})}. \nonumber
 \end{align}
 \end{lemma}

  \begin{proof} For each voter $i\in N$, note that $i\in N_{\mathbf a,a_i}$ and (from the consistency of $\mathbf x$) that $a_i$ is a nearest candidate for $i$.
 If $y\in M\setminus\{a_i\}$, {\em i.e.}, $i\in N\setminus N_{\mathbf a,y}$, \comment{then $a_i\in M_y$.} %and $M_y$ is a nearest committee for $i$.
For every committee $M_y\neq M_{y^*}$, notice that the candidate $y$ belongs to $M_{y^*}$, \comment{giving $d(y,M_{y^*})=0$.} %(giving $d(y,M_{y^*})=0$), and if $i\in N_{\mathbf a,y}$, {\em i.e.}, $a_i=y$, then $d(x_i,M_{y^*})\le d(x_i,y)=d(x_i,a_i)\le d(x_i,M_{y^*})$.
So, the social cost of $M_y$ with $y\ne y^*$ is upper bounded by
\begin{eqnarray*}
&&SC(M_y,\mathbf x) \\
&=&\sum_{i\in N\backslash N_{\mathbf a,y}}d(x_i,M_y)+\sum_{i\in N_{\mathbf a,y}}d(x_i,M_y)\\
&\le& \sum_{i\in N\backslash N_{\mathbf a,y}}d(x_i,M_{y^*})+\sum_{i\in N_{\mathbf a,y}}(d(x_i,y)+d(y,M_y))\\
&=&SC(M_{y^*},\mathbf x)+|N_{\mathbf a,y}|d(y,M_y).
\end{eqnarray*}
Since $2d(x_i,M_{y^*})\ge d(x_i,a_i)+d(x_i,M_{y^*})\ge d(a_i,M_{y^*})$ for every $i\in N$,   the optimal social cost is lower bounded by
\begin{align*}
SC(M_{y^*},\mathbf x)&=\sum_{i\in N}d(x_i,M_{y^*})\\
&\ge \sum_{i\in N}\frac{d(a_i,M_{y^*})}{2}\\
&=\frac 1 2 \sum_{y\in M}|N_{\mathbf a,y}|d(y,M_{y^*})\\
&=\frac 1 2|N_{\mathbf a,y^*}|d({y^*},M_{y^*}).
\end{align*}
The above two bounds give the following estimate on the ratio of the expected social cost of the committee output by the mechanism to the optimum:
\[\begin{array}{rl}
&\frac{\sum_{y\in M}P(M_{y})SC(M_{y},\mathbf x)}{SC(M_{y^*},\mathbf x)}\\
=&P(M_{y^*})+\frac{\sum_{y\in M\setminus\{y^*\}}P(M_y)SC(M_y,\mathbf x)}{SC(M_{y^*},\mathbf x)} \\
\le& P(M_{y^*})+\frac{\sum_{y\in M\setminus\{y^*\}}P(M_y)(SC(M_{y^*},\mathbf x)+|N_{\mathbf a,y}|d(y,M_y))}{SC(M_{y^*},\mathbf x)}\\
\le& 1+\frac{2\sum_{y\in M\setminus\{y^*\}}P(M_y)|N_{\mathbf a,y}|d(y,M_y)}{|N_{\mathbf a,y^*}|d({y^*},M_{y^*})},
\end{array}\]
which proves the lemma.
\end{proof}

A  natural idea to design a mechanism is making the right hand side of inequality (\ref{eq:disss}) as small as possible.
Next, we seek a suitable mechanism whose probabilities of winning set selections achieve this goal.

\vspace{2mm}\hspace{-3mm}\textbf{Mechanism 2 \comment{\sc(Power-Proportionality)}.}   Given an election $\Gamma=(\Omega,M,\mathbf a)$, for every \comment{committee $M_y\in K$}, the winning probability is
\begin{equation}\label{eq:py}
P(M_y)=\frac{{|N_{\mathbf a,y}|^{-m}d(y,M_y)^{-m}}}{\sum_{z\in M}{|N_{\mathbf a,z}|^{-m}d(z,M_z)^{-m}}}.
\end{equation}

\begin{theorem}\label{thm:pow}
\comment{\sc Power-Proportionality} is a randomized scv mechanism that is strategy-proof and has distortion at most $3-\frac{2}{m}$ for social cost objective.
\end{theorem}
\begin{proof}
As $\sum_{y\in M}P(M_y)=1$, the probability distribution is well-defined. To see the strategy-proofness, consider any location profile $\mathbf x$ and an arbitrary voter $i$, one of whose nearest candidates \comment{being} %supposed to be
 $y$. It is easy to see that, $i$ voting for $y$ \comment{(in comparison with not doing so)} increases the size of $N_{\mathbf a,y}$, and decreases the probability $P(M_y)$.
The expected cost of voter $i$ is $P(M_y)d(x_i,M_y)+(1-P(M_y))d(x_i,y)$. Since $d(x_i,y)\le d(x_i,M_y)$, the truth-telling strategy always minimizes her expected cost, which indicates the strategy-proofness.

Next, \comment{we investigate the distortion w.r.t. $\mathbf x\in\chi(\Gamma)$. By Lemma \ref{lem:simi}}, substituting the probability (\ref{eq:py}) into \comment{inequality} (\ref{eq:disss}), we have
\begin{align}
 %\begin{array}{rl}
&\frac{\sum_{y\in M}P(M_{y})SC(M_{y},\mathbf x)}{SC(M_{y^*},\mathbf x)}\label{eq:tran}\\
\le& 1+\frac{2\sum_{y\in M\setminus\{y^*\}}{|N_{\mathbf a,y}|^{1-m}d(y,M_y)^{1-m}}}{|N_{\mathbf a,y^*}|d(y^*,M_{y^*})\sum_{y\in M}{|N_{\mathbf a,y}|^{-m}d(y,M_y)^{-m}}}.\nonumber
%\end{array}
\end{align}
 Now we compute the maximum value of the right hand side in (\ref{eq:tran}) by the function $g:\mathbb R^m_+\rightarrow \mathbb R$, $$g(\alpha_1,\ldots,\alpha_m)=1+\frac{2\alpha_1\sum_{i=2}^m{\alpha_i^{m-1}}}{\sum_{i=1}^m\alpha_i^m}.$$
By the derivative of this function, we know that the maximum value is attained when $\alpha_1=\cdots=\alpha_m$, that is, $\max g(\alpha_1,\ldots,\alpha_m)=g(\alpha_1,\ldots,\alpha_1)=3-\frac{2}{m}$. The right hand side of (\ref{eq:tran}) has the same form as $g$, and it is also at most $3-\frac{2}{m}$, which gives the upper bound of the distortion.
\end{proof}

%--------------------------------------------------------------------------------

\section{Mechanisms for Maximum Social Utility}\label{sec:su}

In this section, we focus on the social utility objective. Each voter targets a favorite candidate, and takes the distance to the eliminated candidate as her utility, as she wants to stay as far away from the nuisance as possible.
%This setting partially coincides with the obnoxious facility game \cite{cheng2011}.\footnote{In that game, the authority wants to build an obnoxious facility on a network, where %the facility is undesirable and all agents try to be stay away from it. At first
 %\comment{the agents report their locations to the authority who accordingly selects a place to locate the facility, aiming to maximizing  the total distance between the agents and the facility}.}

 By %From
 a simple adaptation to the proof of Proposition \ref{prop:poplower1}, one can easily obtain the following lower bounds for both deterministic and randomized mechanisms.
\begin{proposition}\label{prop:lower}
For the social utility objective, no deterministic (resp. randomized) scv mechanism can have a distortion smaller than $3$ (resp. $1.5$).
\end{proposition}

We present in Section \ref{sec:proj} a deterministic svc mechanism with distortion 3, using a dual idea of Mechanism 1. Then, we provide in Sections \ref{sec:prop} and \ref{sec:line} randomized mechanisms for some important special metric spaces.

\subsection{Projection Mechanism}\label{sec:proj}

Recall that the \emph{projection distance} of candidate $y\in M$ on an action profile $\mathbf a=(a_i)_{i\in N}$ is $pd_{\mathbf a}(y)=\sum_{i\in N}d(a_i,y)$. We follow the dual spirit of Mechanism 1 to \comment{select a committee with the eliminated candidate  maximizing the projection distance.}

\vspace{2mm}\hspace{-3.5mm}\textbf{Mechanism 3 \sc(Max-Projection-Distance)}. Given an election $\Gamma=(\Omega,M,\mathbf a)$, the deterministic mechanism $f$ outputs committee $M_y$ where $y$ has the largest projection distance on $\mathbf a$, that is, $y\in \arg\max_{w\in M}pd_{\mathbf a}(w)$ and $f(\mathbf a)=M_y$, breaking ties arbitrarily.

\medskip The following 3-distortion performance guarantee can be proved by an argument that is completely symmetrical with the  proof of Theorem \ref{th:min}.
\begin{theorem}\label{th:max}
{\sc Max-Projection-Distance} is a deterministic {polynomial-time} $3$-distortion scv mechanism for the social utility objective.
\end{theorem}

\iffalse
\begin{proof}
  Let $M_y$ be the committee selected by the mechanism. For any consistent location profile $\mathbf x\in\chi(\Gamma)$, suppose the optimal committee is $M_{y^*}$. The selection rule of the mechanism gives $\sum_{i\in N}d(a_i,y)\ge \sum_{i\in N}d(a_i,y^*)$. It follows that
\begin{eqnarray*}
\frac{SU(M_{y^*})}{SU(M_y)}&=&\frac{\sum_{i\in N}d(x_i,y^*)}{\sum_{i\in N}d(x_i,y)}\\
&\le&\frac{\sum_{i\in N}d(x_i,a_i)+\sum_{i\in N}d(a_i,y^*)}{\sum_{i\in N}d(x_i,y)}\\
&\le& 1+\frac{\sum_{i\in N}d(a_i,y^*)}{\sum_{i\in N}d(x_i,y)}\\
&\le& 1+\frac{2\sum_{i\in N}d(a_i,y^*)}{\sum_{i\in N}d(a_i,y)}\\
&\le& 3,
\end{eqnarray*}
which proves the upper bound 3 on the distortion.
\end{proof}
\fi

This 3-distortion scv mechanism is the best that one can expect for deterministic mechanisms, in view of Proposition \ref{prop:lower}. In contrast to Mechanism 1, it is not strategy-proof: When a voter has two favorite candidates and votes for them respectively, resulting in different action profiles, the corresponding outputs of \comment{\sc Max-Projection-Distance} may be two candidates that have different distances to her. Therefore, to maximize her utility, she has to vote for the specific candidate who leads to a better outcome.

\subsection{Proportionality Mechanism}\label{sec:prop}

A natural idea for randomization is selecting a committee in $K$ with a probability proportional to the number of voters who vote for it. We show the strategy-proofness, and evaluate the distortion in the two-candidate case and  simplex case.

Recall that $N_{\mathbf a,y}=\{i\in N|a_i=y\}$ is the set of \comment{voters who vote for} the candidate $y\in M$. % and we generalize this definition from a single candidate to a subset. Formally, given an action profile $\mathbf a$, for any subset $U\subseteq M$, denote by $U^*$ the set of \comment{voter}s whose actions are in $U$, \emph{i.e.}, $U^*=\{i\in N|a_i\in U\}$.

\vspace{2mm}\hspace{-3.5mm}\textbf{Mechanism 4 \comment{\sc(Proportionality)}.} Given an election $\Gamma=(\Omega,M,\mathbf a)$, for each committee $M_y$, $y\in M$, the winning probability is $$P(M_y)=\frac{n-|N_{\mathbf a,y}|}{(m-1)n}.$$

\smallskip Note that the probability distribution is well-defined, as the sum of $n-|N_{\mathbf a,y}|$ \comment{over} $y\in M$ is $(m-1)n$. %\comment{The following lemma establishes the strategy-proofness.}

\begin{lemma}
\comment{\sc Proportionality} is strategy-proof.
\end{lemma}
\begin{proof}
Consider an arbitrary voter $i\in N$, and suppose \comment{$y\in M$} is her favorite candidate. If \comment{voter} $i$ switches her action from $y$ to any other $y'\in M_y$, then the probability $P(M_y)$ increases,  $P(M_{y'})$ decreases, and  all other probabilities remain the same. The expected utility of voter $i$ is $P(M_y)d(x_i,y)+P(M_{y'})d(x_i,y')+U$ with a fixed value $U$. Since $d(x_i,y)\le d(x_i,y')$,  this implies that the expected utility is non-increasing by switching from $y$ to $y'$. Therefore, being truthful is the optimal strategy, regardless of the actions of other \comment{voter}s.
\end{proof}

By an analysis similar to the proof of Lemma \ref{lem:simi}, we obtain a lower bound on the ratio between the expected social utility of the selection and the optimal utility.

\begin{lemma}\label{lem:alter}
Given \comment{a single-winner election $\Gamma=(\Omega,M,\mathbf a)$} and location profile \comment{$\mathbf x\in\chi(\Gamma)$, suppose $M_{y^*}$} is an optimal committee. For any randomized mechanism that selects $M_y$ ($y\in M$) as winning committee with probability $P(M_y)$, the expected social utility satisfies
\begin{eqnarray*}%\label{eq:bound}
&&\frac{\sum_{y\in M}P(M_y)SU(M_y)}{SU(M_{y^*})}\\
&\ge& 1-\sum_{y\in M_{y^*}}{P(M_y)}\left({1+\frac{\sum_{z\in M}|N_{\mathbf a,z}|d(z,y)}{2(n-|N_{\mathbf a,y^*}|)d(y,y^*)}}\right)^{-1}.
\end{eqnarray*}
\end{lemma}

With the help of %the lower bound presented in
Lemma \ref{lem:alter}, we  {can} upper bound  the distortion of {\sc Proportionality} in the 2-candidate case  \comment{({\em i.e.}, $m=2$)} and simplex case. %\comment{Due to limitation of spaces, the proof is omitted.}
We say the candidates form a {\em simplex}, if the distance between any two candidates is the same, say 2, \emph{i.e.},  $d(y,z)=2$  for all distinct $y,z\in M$.%The setting corresponds to the case in which candidates are,   in a sense, completely symmetrical.  %share no similarities, i.e., when all candidates are equally different from each other.
%It models the scenario that the authority has no non-trivial information on the candidates and selects the \comment{winner} %winners depending only on the ballots.
\footnote{The simplex is studied in \cite{anshelevich2017randomized} %.\footnote{\comment{Using the idea of proportional to squares, Anshelevich and Postl  \cite{anshelevich2017randomized} design a 2-distortion randomized mechanism   for minimizing social cost in
\comment{for single-winner election, under some additional assumption on distances.}} % the additional assumption that %but it is additionally required that the distance from each voter to any candidate is no more than 1.}

\begin{theorem}\label{thm:pmech}
For the social utility objective, {\sc Proportionality}  has distortion\\
(i)  at most {$1.523$} when $m=2$;\\
(ii)  at most $3-\frac{4}{m+2}$ {when candidates form a simplex}.
\end{theorem}

%It is worth noting that in the special case of $m=2$, the two candidates trivially form a simplex. The distortion upper bound  in the case $m=2$ \comment{in (ii) is  larger than  the ratio 1.523 in (i). The loss in the estimation stems from the analysis in (ii)}, which needs to be more general to take care of all possible number $m$ of candidates.

%We \comment{break the proof of Theorem \ref{thm:pmech} into the following} four lemmas: the first one \comment{establishes} the strategy-proofness, and others concern the distortion.

\begin{proof}
(i) For any election $\Gamma$ and consistent location profile $\mathbf x\in\chi(\Gamma)$, suppose $y$ is the optimal candidate (singleton committee), and \comment{$y^*$} is the other one. By Lemma \ref{lem:alter}, we have
\begin{eqnarray*}
&&\frac{P(y)SU(y)+P(y^*)SU(y^*)}{SU(y)}\\
&\ge& 1-P(y^*)\left({1+\frac{|N_{\mathbf a,y^*}|}{2(n-|N_{\mathbf a,y^*}|)}}\right)^{-1}\\
&=&1-\frac{|N_{\mathbf a,y^*}|}{n}\left({1+\frac{|N_{\mathbf a,y^*}|}{2(n-|N_{\mathbf a,y^*}|)}}\right)^{-1}\\
&=&1-\left({\frac{n}{|N_{\mathbf a,y^*}|}+\frac{n}{2(n-|N_{\mathbf a,y^*}|)}}\right)^{-1}\\
&\ge& 1-\left(1.5+\sqrt2\right)^{-1}%1-\frac{2}{3+2\sqrt{2}}=4\sqrt{2}-5.
\end{eqnarray*}
Therefore, the distortion is at most $(1-(1.5+\sqrt2)^{-1})^{-1}=(5+4\sqrt2)/7=1.5224...$

(ii) can be proved in a similar but more involved analysis, which is relegated to Supplementary Material.
\end{proof}

%-------------------------------------------------------------------------------------

\subsection{Mechanisms on the Real Line}\label{sec:line}

We now consider the case where all voters and candidates are located on the real line, and the metric is defined as the Euclidean distance. This setting simulates %corresponds to
the scenario in which an authority wants to build a facility on a street, and has been extensively studied for obnoxious facility games. %It can be assumed w.l.o.g. that the leftmost (resp. rightmost) candidate is located at $y_1=0$ (resp. $y_m=L$).
The results of \cite{cheng2011} implies that  an optimal committee must eliminate one of the two endpoints of the line segment spanned by $y_1,\ldots,y_m$. This nice fact directly provides a randomized strategy-proof 2-distortion mechanism that eliminates the leftmost candidate and the rightmost candidate with probability $\frac 1 2$, respectively. Next, we improve the distortion by a more involved probability distribution of selection, at a cost of losing the strategy-proofness.

\vspace{2mm}\hspace{-3.5mm}\textbf{Mechanism 5 \comment{\sc(Left-or-Right)}.} Given an election \comment{$\Gamma=(\mathbb R,M,\mathbf a)$}, where the leftmost and rightmost candidate are located at \comment{$y_1=0$ and $y_m=L$, respectively}. Denote by $n_1,n_2$ the number of voters whose actions are on $[0,\frac L 2], (\frac L 2, L]$ , respectively. Select $M_{y_i}$ with probability $P(M_{y_i})$, $i=1,m$, as specified below:
\begin{itemize}
\item[$\bullet$] If $n_1>n_2$, then $P(M_{y_1})=\frac 6 {13}$ and $P(M_{y_m})=\frac 7 {13}$.
\item[$\bullet$] If $n_1<n_2$, then $P(M_{y_1})=\frac7 {13}$ and $P(M_{y_m})=\frac 6 {13}$.
\item[$\bullet$] If $n_1=n_2$, then $P(M_{y_1})=P(M_{y_m})=\frac 1 2$.
\end{itemize}

\begin{theorem}\label{thm:line}
{\sc Left-or-Right} is a randomized $\frac{13}{7}$-distortion scv mechanism for the social utility objective.
\end{theorem}

\begin{proof}
For any election $\Gamma=(\mathbb R,M,\mathbf a)$ and consistent location profile \comment{$\mathbf x\in\chi(\Gamma)$}, we show that the  \comment{performance ratio $\frac{OPT_u(\mathbf x)}{\mathbf E[SU(f(\mathbf a),\mathbf x)]}$} is upper bounded by $\frac{13}{7}$\comment{, where $f$ denotes mechanism {\sc Left-or-Right}}. It is easy to see that the worst case w.r.t. the performance ratio must occur when all voters are also located on interval $[0,L]$. \comment{(If some $x_i$ is smaller than 0 or larger than $L$, then changing it to 0 or $L$ would not decrease the ratio.) So we assume that} $x_i\in[0,L]$ \comment{for all $ i\in N$, and} only consider the line segment $[0,L]$.

If $n_1=n_2$, then the expected social utility of the outcome is
\begin{align*}
\begin{array}{rcl}
\mathbf E[SU(f(\mathbf a))]&=&\frac 1 2SU(M_{y_1})+\frac 1 2SU(M_{y_m})\vspace{1mm}\\
&=&\frac 1 2\sum_{i=1}^nx_i+\frac 1 2\sum_{i=1}^n(L-x_i)=\frac{Ln}{2},
\end{array}
\end{align*}
and the optimal social utility is $OPT_u(\mathbf x)=\max\{SU(M_{y_1}),SU(M_{y_m})\}$.  Since $n_1=n_2$, we have
\[\begin{array}{c}SU(M_{y_1})=\sum_{i=1}^nx_i\le \frac 3 4Ln_1+Ln_2=\frac{7Ln}{8},\end{array}\] where $SU(M_{y_1})$ reaches the upper bound when $n_1$ voters who vote for the midpoint candidate $\frac L 2$ are located at $\frac{3L}{4}$, and $n_2$ voters who vote for $y_m=L$ are located at $L$. Similarly, we have \[\begin{array}{c}SU(M_{y_m})\le \frac{7Ln}{8}.\end{array}\]
Therefore, we have $OPT_u(\mathbf x)\le \frac 7 4\mathbf E[SU(f(\mathbf a))]<\frac {13}7\mathbf E[SU(f(\mathbf a))]$ as desired.

  When $n_1\neq n_2$, by symmetry, we only discuss the case $n_1>n_2$. Recall that an optimal solution eliminates either $y_1$ or $y_m$. First, if the optimal committee is $M_{y_1}$, we also have $OPT_u(\mathbf x)=\sum_{i=1}^nx_i\le \frac 3 4Ln_1+Ln_2$ by the same reasoning as above. In turn, $\frac34n_1+n_2=n-\frac{n_1}4<\frac78n$ gives  $OPT_u(\mathbf x)<\frac{7Ln}{8}$. It follows that
\begin{align*}
\begin{array}{rcl}
\mathbf E[SU(f(\mathbf a))]&=&\frac{6}{13}\sum_{i=1}^nx_i+\frac{7}{13}\sum_{i=1}^n(L-x_i)\vspace{1mm}\\
&=&\frac{7}{13}Ln-\frac{1}{13}\sum_{i=1}^nx_i\nonumber\vspace{1mm}\\
&>& \frac{8}{13}OPT_u(\mathbf x)-\frac{1}{13}OPT_u(\mathbf x)\vspace{1mm}\\
&=&\frac{7}{13}OPT_u(\mathbf x).
\end{array}
\end{align*}

Next, if the optimal committee is $M_{y_m}$, with $\epsilon>0$ being infinitesimal, we have \[\begin{array}{c}OPT_u(\mathbf x)=\sum_{i=1}^n(L-x_i)\le Ln_1+\frac 3 4Ln_2-\epsilon< Ln,\end{array}\]
 where the first inequality  holds with equality when $n_2$ voters who vote for \comment{$\frac L 2+\epsilon'$} are located at \comment{$\frac{L}{4}+\epsilon'$ ($\epsilon'>0$ being infinitesimal)}, and $n_1$ \comment{voter}s who vote for $y_1=0$ are located at $0$. Therefore, % The expected social utility of the outcome is
\begin{align*}
\begin{array}{rcl}
\mathbf E[SU(f(\mathbf a))]&=&\frac{6}{13}\sum_{i=1}^nx_i+\frac{7}{13}\sum_{i=1}^n(L-x_i)\vspace{1mm}\\
&=&\frac{6}{13}Ln+\frac{1}{13}\sum_{i=1}^n(L-x_i)\nonumber\vspace{1mm}\\
&>& \frac{6}{13}OPT_u(\mathbf x)+\frac{1}{13}OPT_u(\mathbf x)\vspace{1mm}\\
&=&\frac{7}{13}OPT_u(\mathbf x).%\label{eq:second}
\end{array}\end{align*}
\comment{The proof is complete.}%Combining the two cases $n_1=n_2$ and $n_1\neq n_2$, we complete the proof.
\end{proof}

%Even though \comment{\sc Left-or-Right} has a better distortion than the trivial 2-approximate one, it does not satisfy the strategy-proofness: Consider an election where three candidates locate at $0,\frac L 2$ and $L$, two voters located at $\frac 3 5L$ and $L$ vote for $\frac L 2$ and $L$, respectively. Clearly, the location profile is consistent with the election, {\em i.e.}, both voters vote their nearest candidate. Since $n_1=n_2=1$, the mechanism {eliminates} the two endpoints with probability $\frac 1 2$, respectively. However, if the \comment{voter} at $\frac 3 5L$ votes for the candidate at $L$, the mechanism {eliminates} the left endpoint $0$ with a higher probability, and thus this voter receives a higher expected utility.

 %------------------------------------------------------------------------

\section{Concluding Remarks}\label{sec:con}\label{sec:add}

In this paper we are concerned with the \emph{scv} mechanisms for single-loser election, instead of \emph{ranking} mechanisms that ask the ordinal preferences of voters.
We study  how well, in terms of minimizing (maximizing) social cost (utility), the mechanisms that only receive the information on top-ranked candidates can compete with omniscient selection rules. From the worst-case perspective, our results show that accessing the very limited information is often enough, in view that the performance guarantees of the mechanisms we propose match the lower bounds which hold even when ranking preferences are known.

\paragraph{Extension.} The good performances of scv mechanisms can be extended to a more general task: %suppose we no longer plan to eliminate one candidate, but aim at
selecting a size-$k$ committee $W$ as winners for a predetermined integer $k\le m-1$.
The voters may take the distance to the winners' set $W$ as their costs, or take the distance to the losers' set $M\backslash W$ as their utilities. For the social cost (utility) objective, a couple of ideas and results presented in Sections \ref{sec:sc} and \ref{sec:su} can be generalized. Specifically, we obtain the following lower bounds (LB) and upper bounds (UB) on the distortions of scv mechanisms for selecting a size-$k$ committee:
\begin{table}[htbp]
  \centering
  %\caption{The lower bound (LB) and upper bound (UB) on the distortions of mechanisms for selecting a size-$k$ committee}\label{table}

    \begin{tabular}{|p{2cm}<{\centering}|p{2.12cm}<{\centering}|p{2cm}<{\centering}|}
    \hline
    \small Objective &  \small Deterministic &\small Randomized  \\
    \hline
    {\small Social cost} & {\small ${\rm LB}={\rm UB}=3$ } & {\small LB: $2$ } \\

                                       \hline
   {\small{Social utility}} & {\small ${\rm LB}={\rm UB}=3$ } & {\small LB: 1.5}\\
    \hline
    \end{tabular}
  \label{tab:1}
\end{table}

 \noindent The upper bound  3 for the social cost (utility) objective is guaranteed by outputting a size-$k$ committee that minimizes the projection distance (whose complement set maximizes the projection distance). More details could be found in Supplementary Material.

\paragraph{Future direction.} Although strategy-proof mechanisms for the single-winner and single-loser voting have been explored more or less, for the general problem of selecting a size-$k$  committee by scv rules, so far, to the best of our knowledge, there is no performance-guaranteed mechanism that is strategy-proof, even for $k=2$. This suggests an interesting research direction for scv mechanism design. %This would be an interesting topic for future research.
Except for the proportional idea employed by Mechanism 2 and 4, the quadratic proportionality \cite{meir2012algorithms,anshelevich2017randomized} or other proportional probabilities relying on $k$ may be useful.

\bibliography{reference}

\end{document}